\documentclass{llncs}

\usepackage{amssymb}            
\usepackage{amsmath,mathtools}  
\usepackage{tikz} 
\usetikzlibrary{chains,positioning,scopes,arrows,plotmarks} 
\usepackage{comment}                                
\usepackage{array}                              
\usepackage{xspace}
\usepackage{color}
\usepackage[utf8]{inputenc}
\usepackage{psfrag}

\newcommand{\ourgame}{celebrity game\xspace}
\newcommand{\ourgames}{celebrity games\xspace}
\newcommand{\Ourgames}{Celebrity  games\xspace}
\newcommand{\ourgamen}{star celebrity game\xspace}
\newcommand{\ourgamens}{star celebrity games\xspace}

\newcommand{\stargame}{star celebrity game\xspace}
\newcommand{\stargames}{star celebrity games\xspace}

\newcommand{\maxG}{Max game\xspace}

\newcommand{\sumG}{Sum game\xspace}

\newcommand{\maxbd}{MaxBD game\xspace}
\newcommand{\maxbds}{MaxBD\ games\xspace}
\newcommand{\sumbd}{SumBD game\xspace}

\newcommand{\cmaxbd}{c^\textsf{MaxBD}}
\newcommand{\Cmaxbd}{C^\textsf{MaxBD}}

\newcommand{\optv}{\textsf{opt}\xspace}
\newcommand{\OPT}{\textsc{opt}\xspace}

\newcommand{\NE}{\textsc{ne}\xspace}

\newcommand{\NES}{\text{NE}\xspace}
\newcommand{\NCG}{\text{NCG}\xspace}

\newcommand{\PoA}{\text{PoA}\xspace}
\newcommand{\PoS}{\text{PoS}\xspace}
\newcommand{\NP}{\textsf{NP}\xspace}
\newcommand{\wM}{w_{max}}
\newcommand{\wm}{w_{min}}
\newcommand{\remove}[1]{}

\author{C. \`Alvarez \and M. Blesa \and A. Duch \and A. Messegu\'e \and M. Serna}
\institute{ALBCOM Research Group, Computer Science Department, UPC, Barcelona\\
\email{\{alvarez,mjblesa,duch,messegue,mjserna\}@cs.upc.edu}}

\title{Stars and Celebrities: A Network  Creation Game}

\begin{document}

\maketitle
    
\begin{abstract}

We introduce  \emph{\Ourgames}, a new model of network creation games. 
In this model players have weights (being $W$ the sum of all the player's weights) 
and there is a critical distance $\beta$ as well as a link cost $\alpha$.  
The cost incurred by a player depends on the cost of establishing links 
to other players and on the sum of the weights of those players that remain 
farther than the critical distance. 
Intuitively,  the aim of any player is to be relatively close 
(at a distance less than $\beta$) from the rest of players, 
mainly of those having high weights. 
The main features of celebrity games are that: 
computing the best response of a player is NP-hard if $\beta> 1$ 
and polynomial time solvable otherwise;
they always have a pure Nash equilibrium; the family of celebrity
games having a  connected Nash equilibrium  is characterized 
(the so called {\em star celebrity games}) and bounds on the 
diameter of the resulting equilibrium graphs are given; 
a special case of star celebrity games share its set of Nash 
equilibrium profiles with the MaxBD games with uniform 
bounded distance $\beta$ intoduced in \cite{BiloGP:12}. 
Moreover, we analyze the Price of Anarchy (\PoA) and of Stability (\PoS) 
of celebrity games and give several bounds. These are that: for non-star celebrity games
$\PoA = \PoS = max\{1,W/\alpha\}$; for star celebrity games $\PoS = 1$
and $\PoA = O(min\{n/\beta, W\alpha\})$ but if the Nash Equilibrium is a 
tree then the \PoA is $O(1)$; finally, when $\beta=1$ the \PoA is at most 2. 
The upper bounds on the \PoA are complemented with some lower bounds for $\beta=2$.  
\end{abstract} 
 




%
\section{Introduction}

The global growth of Internet and social networks 
usage has been accompanied by an increasing 
interest to model theoretically their creation as 
well as their behavior. 
In particular, network creation games (\NCG)
aim to model Social Networks and Internet 
by simulating the creation of a decentralized 
and non-cooperative communication network 
among $n$ players (the network nodes).

From the seminal paper~\cite{Fe:03} 
several proposals have been made in the area of \NCG.
In the original model, 
the goal of each player is to have, 
in the resulting network, 
all the other nodes as close as possible 
while buying as few links as possible~\cite{Fe:03}. 
Several assumptions are made: all the players 
have the same interest 
(all-to-all communication pattern with identical weights);
the cost of being disconnected is infinite; 
and the edges paid by one node can be used by others. 
Formally, a game $\Gamma$ in this model  is defined as a tuple $\Gamma=\langle  V, \alpha\rangle$,  
where $V$ is the set of $n$ nodes and $\alpha$ the cost 
of establishing a link. 
A strategy for player $u\in V$ is a subset $S_u\subseteq V-\{u\}$, 
the set of players for which player $u$ pays for establishing a link.
The $n$ players and their joint strategy choices $S=(S_u)_{u\in V}$ 
create an undirected graph $G[S]$.
The cost function for each node
$u$ under strategy $S$ is defined by 
$c_u(S) = \alpha|S_u|+\sum_{v\in V}d_{G[s]}(u,v)$ 
where $d_{G[s]}(u,v)$ is the distance between nodes $u$ and $v$ 
in graph $G[S]$. Because of the summation in the 
cost function this model is informally 
known as the \emph{\sumG} model. 
By changing the cost function to $c_u(S)=\alpha |s_u| + \max \{d_{G[S]}(u,v) | v\in V\}$
as proposed in \cite{Demaineetal:07-12} one obtains the 
\emph{\maxG} model. 

From here on several versions and variants have been considered. 
Instead of buying links unilaterally, 
\cite{CorboParkes:05} proposed the possibility of having links 
formed by bilateral contracting: both endpoints must agree before 
creating a link between them and the two players share (half-half) the cost 
of establishing the link. 
\NCG models can be cooperative --a possibility 
introduced by \cite{Albersetal:06}-- and therefore 
any node can purchase any amount of any link in the resulting graph, 
and a link can be created when its cost is covered by a set of players.
The model studied in \cite{BiloGP:15} (see also  \cite{BiloGP:12})  
considers the notion of bounded distance 
per player and propose  two variants: the \maxbd and the \sumbd, 
corresponding to the original Max and Sum cost models respectively.   
The cost in those games depends on whether 
the player's eccentricity is smaller or equal than 
the associated bounded distance. 
In that case a player pays  the number of established
links, otherwise its cost is  infinite.
For further variants we refer the 
  Interested reader 
to~\cite{Demaineetal:07-12,LeonardiSankowski:07,BrandesHN:08,DemaineHMZ:09,Lenzner:11,AlonDHL:13,AlonDHK:14,BiloGLP:15,NikolePRS:15,CordL:15,Ehsani:2015} among others.

We introduce  \emph{\ourgames} a \NCG  
where players have different weights and 
share a common distance bound. As far as we understand, not all the nodes 
in Internet based networks have the same importance. 
It is though natural to consider players with different relevance weights. 
In such a setting, the cost of being far (even if connected) 
from important nodes (the ones with high weight)
should be higher than the cost of having them close. 
Intuitively, the goal of each player  in \ourgames is to 
buy as few links as possible in order to 
have the high-weighted nodes (or groups of nodes) 
closer to the given critical distance. 
Observe that if the cost of establishing links is higher 
than the benefit of having close a node (or set of nodes),  
players might rather  prefer to stay either far or even 
disconnected from it.

Our aim is to study 
the combined effect of having players with different weights 
that share a common bounded distance. 
Although  heterogeneous players have 
been considered recently in the context of \NCG 
under bilateral contracting ~\cite{MMO-EC14,AlvarezSF:15},
and  \cite{BiloGP:15} 
consider the notion of bounded distance, 
to the best of our knowledge this is the first 
model that studies how  a  common critical distance, 
different  weights, and a link cost, altogether 
affect the individual preferences of the players.
 
In our model the cost of a player has two components.  
The first one is the cost of the 
links established by the node.      
The second one is the sum of the weights 
of those nodes that are farther away than the critical distance. 
More specifically, the parameters of a \ourgame are:  
a weight to each player;
a cost for establishing a link; and  a critical distance.
Formally,  a {\ourgame\ } is defined by 
$\Gamma= \langle  V, (w_u)_{u\in V}, \alpha,\beta \rangle$,  
where $V$ is  a set of nodes with weights $(w_u)_{u\in V}$, 
$\alpha$ is the cost of establishing a link and, 
$\beta$ establishes the desirable  distance bound. 
\Ourgames include  the \maxbds introduced 
in \cite{BiloGP:15} (see Section~\ref{sec:MaxBD} for the details). 
They capture not only the cases in which players are  indistinguishable  
but those cases  where  the players may have different weights  affecting 
differently  the costs of  the other players.

We analyze  the structural 
properties of the Nash equilibrium (\NE) graphs of \ourgames 
and their quality with respect to the optimal strategies 
under the usual social cost.  To do so we address the cases 
$\beta =1$ and $\beta >1$ separately.  
Notice that, for $\beta = 1$, each player  $u$ has to decide  
for every non-edge $(u,v)$ of the graph  to pay either $\alpha$ for the link   
or  $w_v$ (the weight of the non-adjacent node $v$)  
while, for $\beta>1$, every player $u$ has to choose for 
each non-edge $(u,v)$ between buying the link $(u,v)$ 
and paying  $\alpha$ minus  the sum of the weights of those 
nodes  whose distance to $u$ will become  less or equal 
than the critical distance $\beta$ or paying  the sum of  the 
weights of the nodes with distance to $u$ greater than $\beta$.

For the general case  $\beta>1$ our results can be summarized as follows:  
\begin{itemize}
\item Computing a best response for a player is  \NP-hard
\item The optimal social cost of a celebrity game $\Gamma$  
depends on the relation between the total sum of the weights $W$ and 
the cost $\alpha$ of buying a link: $\OPT(\Gamma)=\min \{\alpha,W\} (n-1)$.
Nevertheless, pure  \NE always exist and 
\NE graphs are either connected or  a  set of isolated nodes. 
Again, the relationship between the cost of establishing a link 
and the weight of the nodes leads to different types of \NE. 

\item We use the term  \emph{celebrity}  for a node whose weight 
is strictly greater than the cost of establishing a link.  
Having at least one celebrity guarantees that all \NE graphs
are connected, although there are \ourgames without celebrities 
that still have connected \NE graphs.  
In those games having a connected 
\NE graph, a star tree  is always a \NE graph. 
We called  this  subfamily of \ourgames \emph{\stargames}. 

\item For \stargames, we obtain a general upper bound of 
$2 \beta +1$ for the diameter of \NE graphs. 
In particular, if $G$ is a  \NE tree we show that $diam(G) \leq \beta +1$, 
otherwise $\beta / 2 < diam(G) \leq 2 \beta +1$.
The upper bound can be improved by considering the relationship 
between $\alpha$ and the maximum and minimum weights, 
$w_{max}$ and $w_{min}$, respectively.
So, if $w_{min} \leq \alpha < w_{max}$, then $diam(G) \leq 2 \beta$. 
On the contrary, if $\alpha < w_{min}$, then $diam(G)\leq \beta$.

\item  For \stargames with $\alpha < w_{min}$,  
we show that the set of  \NE strategy profiles coincides with   
the set of \NE strategy profiles  of  a \maxbd  with uniform bounded distance $\beta$.

\item We find several bounds on the Price of Anarchy (\PoA) and of stability (\PoS).  
For non-\stargames \PoS = \PoA$ =\max{\{1, W/\alpha\}}$.  
For \stargames the \PoS is 1 and
we obtain a general upper bound of 
$O(\min\{n/\beta,W/\alpha\})$ for the \PoA. We also show particular games 
on $n$ players having $\PoA=\Omega(n)$, for $\beta=2$.  
To complement those results we prove that 
the \PoA on \NE trees is constant (special cases like trees 
are also considered in the literature, see for instance \cite{AlonDHL:13,AlonDHK:14,Ehsani:2015}.
\end{itemize}
Finally, for the  particular case $\beta=1$, 
we show that computing a best response for a player  is polynomial time solvable 
and that the \PoA is at most  2.

\smallskip

The paper is organized as follows. 
In Section~\ref{sec:model} we introduce the basic definitions and the  \ourgames model. 
We also show that computing a  best response is $\NP$-hard.
In Section \ref{sec:so-and-ne}  we set the fundamental properties of 
\NE and optimal  graphs. We characterize \stargames and we provide the first bounds 
for the \PoA and the \PoS.  
Section~\ref{sec:ne-diameter} is devoted to the study 
of the diameter of \NE graphs. Section~\ref{sec:MaxBD} 
studies the relation between the \maxbd model 
and the \ourgame model. In Section~\ref{sec:poa-lower-bound} we derive 
the bounds for the \PoA.  In Section~\ref{sec:poa-trees} 
we give the upper bound of the \PoA over \NE trees 
and  in Section~\ref{sec:beta1} we study the case $\beta=1$. 
Finally, we state some conclusions and 
open problems in Section~\ref{sec:conclu}.

\section{The Model}
\label{sec:model}
In this section we introduce  \ourgames\ 
 and we analyze the complexity  of computing a  best response. 
Let us start with some definitions. 
We use standard notation for graphs and strategic games. 
All the graphs in the paper are undirected unless explicitly said otherwise. 
Given a graph $G=(V,E)$ and $u,v\in V$, $d_G(u,v)$ denotes 
the \emph{distance} between $u$ and $v$ in $G$, 
i.e., the length of the shortest path from $u$ to $v$. 
The \emph{diameter} (or  \emph{eccentricity}) of a 
vertex $u\in V$ is $diam(u)=\max_{v\in V} d_G(u,v)$ 
and the \emph{diameter} of $G$ is $diam(G)=\max_{v\in V} diam(v)$.  
An \emph{orientation} of  an undirected graph is an assignment 
of a direction to every edge of the graph, turning it into a directed graph.
A \emph{bridge} is an edge whose deletion increases the number of 
connected components of  the graph.
For a  weighted set   $(V, (w_u)_{u\in V})$ we extend the weight function 
to subsets in the usual way.  For $U\subseteq V$, $w(U) = \sum_{u\in U} w_u$. 
Furthermore, we set  $W=w(V)$, $\wM= \max_{u\in V} w_u$ 
and $\wm= \min_{u\in V} w_u$.

\begin{definition}
A \emph{\ourgame} $\Gamma$ 
is a tuple $\langle  V, (w_u)_{u\in V}, \alpha,\beta \rangle$ where:
$V=\{1,\dots, n\}$ is the  set of players, for each player $u\in V$, $w_u> 0$ 
is the \emph{weight} of player $u$, 
$\alpha > 0$ is the cost of establishing a  link, and 
$\beta$, $1\leq \beta\leq n-1$,  is the \emph{critical distance}. 
 
A  \emph{strategy} for player $u$ in $\Gamma$ is a subset $S_u \subseteq V-\{u\}$,  
the set of players for which player $u$  pays  for establishing a direct link.  
A \emph{strategy profile} for $\Gamma$  is a tuple $S = (S_1,\ldots, S_n)$ 
that assigns a strategy to each player.
Every  strategy profile $S$ has associated  an \emph{outcome graph}, 
the undirected graph defined  by $G[S]= (V, \{\{u,v\}| u\in S_v \vee v\in S_u\})$. 

We denote by $c_u(S) = \alpha|S_u| + \sum_{\{v \mid d_{G[S]}(u,v)>\beta\}}w_v$
the \emph{cost} of player $u$ in the strategy profile  $S$.   
And, as usual, the \emph{social cost} of a strategy profile $S$  
in $\Gamma$ is defined as $C(S)=\sum_{u\in V} c_u(S)$.
\end{definition}

Observe that,  even though  a link might be established by only 
one of the two players, we assume that once a link is 
established it can be used in both directions.
Note also that players may have different  weights. 
The player's cost function has two components: 
the cost of establishing links and  the sum of the weights of 
those players who are farther away than the critical distance $\beta$.
In our model links have uniform length therefore w.l.o.g  $\beta$ is an integer. 
In what follows we assume that,
for a celebrity game $\Gamma=\langle  V, (w_u)_{u\in V}, \alpha,\beta \rangle$, 
the parameters verify the required conditions. 
Furthermore, unless specifically stated, we assume $\beta>1$, 
the case $\beta=1$ will be analyzed in Section~\ref{sec:beta1}.
We use the following notation $n=|V|$,  
$\mathcal{S}(u)$ is the set of strategies for player $u$ and 
$\mathcal{S}(\Gamma)$ is the set of strategy profiles of $\Gamma$.
For a strategy profile $S\in \mathcal{S}(\Gamma)$ and a 
strategy $S'_u\in \mathcal{S}(u)$, for player $u$, 
$(S_{-u},S'_u)$ represents the strategy profile in which $S_u$ is replaced by $S'_u$ 
while the strategies of the other players remain unchanged.  
The \emph{cost difference} $\Delta(S_{-u},S'_u)$ is defined as
$\Delta(S_{-u},S'_u)= c_u(S_{-u},S'_u) - c_u(S)$.  
Observe that, if  $\Delta(S_{-u},S'_u) <0$, then player $u$ 
has an incentive to deviate from $S_u$ and select $S'_u$. 
A \emph{best response} to $S\in \mathcal{S}(\Gamma)$ for player $u$ 
is a strategy  $S'_u\in \mathcal{S}(u)$ minimizing $\Delta(S_{-u},S'_u)$.

Let us recall the definition of Nash equilibrium.
\begin{definition}  
Let $\Gamma=\langle  V, (w_u)_{u\in V}, \alpha,\beta \rangle$ be a \ourgame. 
A strategy profile $S \in \mathcal{S}(\Gamma)$  is a \emph{Nash equilibrium} of 
$\Gamma$ if  no player  has an  incentive to deviate from his strategy. 
Formally, for each player $u$ and each strategy  $S_u'\in\mathcal{S}(u) $,   
$\Delta(S_{-u},S_u')  \geq 0$.
\end{definition}

We denote by $\NES(\Gamma)$ the set of Nash equilibria of a game $\Gamma$ 
and we use the term \NE to refer to a strategy profile $S\in \NES(\Gamma)$. 
We say that a graph $G$ is a \emph{\NE graph}  of $\Gamma$ 
if  there is $S\in \NES(\Gamma)$  so that $G=G[S]$.
We will drop the explicit reference to $\Gamma$ whenever $\Gamma$ 
is clear from the context.  
It is worth observing that, for $S\in \NES(\Gamma)$,  
it never happens that $v\in S_u$ and $u\in S_v$, for any $u,v\in V$. 
Thus, if $G$ is the outcome of a \NE  $S$,  
$S$ corresponds to an orientation of the edges in $G$. 
Furthermore, a \NE graph $G$ can be the outcome of several strategy 
profiles but not all the orientations of a \NE graph $G$  are \NE.

Let $\optv(\Gamma)= \min_{S \in \mathcal{S}(\Gamma)} C(S)$  
be the minimum value of the social cost. 
We use the term \OPT  strategy profile to refer to one strategy profile  with optimal social cost. 

Observe that, when in a strategy profile $S$, two players $u$ and $v$ 
are such that $u\in S_v$ and $v\in S_u$,  the social cost is higher  
than when only  one of them is paying for the connection $\{u,v\}$ 
and therefore, as for \NE, this does not happen in an \OPT strategy profile. 
In the following, as we are interested in \NE and \OPT strategies,  
among all the possible strategy profiles having the same outcome graph, 
we only consider those corresponding to orientations of the outcome graph.   
In this sense 
the social cost depends only on  the outcome graph, the weights and the parameters.  
Thus, we can express the social cost of a strategy profile  as a function 
of the outcome graph $G$ as follows
$$\displaystyle C(G)=\alpha |E(G)| + 
\sum_{u\in V}\sum_{\{v \mid d_G(u,v)>\beta\}} \hskip -12pt w_v = 
\alpha |E(G)| + \sum_{\{(u,v)\mid u<v  \text{ and }  d_G(u,v)>\beta\}}  \hskip -12pt (w_u+w_v).$$

We make use of three particular outcome graphs on $n$ vertexes: 
$K_n$, the complete graph; $I_n$, the independent set;  
and $ST_n$ the star graph, i.e., a tree  in which one of the vertexes, 
the \emph{central}  one,  has a direct link to all the other $n-1$ vertexes. 
For those graphs, we have the following values of the social cost.
For  $\Gamma=\langle  V, (w_u)_{u\in V}, \alpha,\beta \rangle$, 
with $|V|=n$,  $C(K_n) =\alpha n(n-1)/2$, $C(I_n) =  W(n-1)$, for $\beta\geq 1$.  Furthermore,  
$C(ST_n)= \alpha (n-1)$, for $1<\beta \leq n-1$,   
and  $C(ST_n)= \alpha (n-1) + (n-2) (W-w_c)$ where $c$ is the central vertex, for $\beta=1$. 

We define  the \PoA and the \PoS as usual.
\begin{definition} Let $\Gamma$ be a celebrity game. 
The \emph{Price of Anarchy} of~$\Gamma$ is defined 
as  $\PoA(\Gamma)= {\max_{S \in  \NES(\Gamma)} C(S)}/{\optv(\Gamma)}$ and
the \emph{Price of Stability} of~$\Gamma$ as
$\PoS(\Gamma)= {\min_{S \in \NES(\Gamma)} C(S)}/{\optv(\Gamma)}$.
\end{definition}

Our first result  shows that  computing  a best response 
in \ourgames is \NP-hard  by a reduction from the minimum dominating set problem. 
The problem becomes tractable for $\beta=1$ as we show in Section~\ref{sec:beta1}.
\begin{proposition}
\label{prop:Best-Response}   
Computing a best response for a player to a strategy profile  in  a  \ourgame 
is  \NP-hard, even when $\beta=2$ and restricted to the cases in which all players 
except possibly one have weights bigger than $\alpha$.
\end{proposition}

\begin{proof}

We provide a reduction from the problem of computing a dominating 
set of minimum size which is a classical \NP-hard problem. 
Recall that a dominating set of a graph $G=(V,E)$ is a set $U\subset V$ 
such that any vertex $u\in V$ is in $U$ or has a neighbor in $U$.

Let $G=(V,E)$ be a graph, we associate to $G$ and $u$ a \ourgame  
$\Gamma  =\langle  V', (w_v)_{v\in V'}, \alpha,\beta \rangle$, and a strategy profile $S$ as follows: 
\begin{itemize}
\item The set of players is $V'= V \cup \{u\}$, where $u$ is a new player (i.e. $u\not \in V$).
\item $\beta=2$, $\alpha =1.5$, 
\item  for every $v \in V$, $w_v = 2$.
\item The strategy profile $S$ is obtained from an orientation of the edges 
in $G$ setting $S_u=\emptyset$.
Observe that by construction $G[S]$ is the disjoint union of $G$ with the isolated vertex $u$.  
\end{itemize}
Finally, set $u$ to be the player for which we want to compute the best response to $S$. 
Observe that the weight of $u$ has not been defined yet.

Let $D \subseteq V$ be a strategy for player $u$. Notice that, if $D$ is a dominating set of $G$, then 
$c_u(S_{-u},D)= \alpha |D| + \sum_{x\in V, d(u,x)>2} 2 = \alpha |D|$.
If $D$ is not a dominating set of $G$,  $c_u(S_{-u},D)= \alpha |D| + \sum_{x\in V, d(u,x)>2} 2 > \alpha (|D| + |\{x\in V |d(u,x)>2\}|$. Then, $D \cup \{x\in V | d(u,x)>2\}$ is a 
better response than $D$ and furthermore it is a dominating set. 
Hence, the best response of player $u$ is a dominating set $D$ of $G$ of minimum size. 
To conclude the proof just  notice that the described reduction is polynomial 
time computable and that we did not make any assumption on the weight of the node $u$. 
\end{proof}
\section{Social Optimum and Nash equilibrium}
\label{sec:so-and-ne}
We analyze here  the main properties of \OPT and \NE  strategy 
profiles  in  \ourgames.
We start analyzing the cost of optimal graphs for the social cost. 
Then we characterize the family of  \ourgamens having a connected \NE graph. 
Finally, we provide  exact bounds on the \PoA and  the \PoS in some particular cases.

\begin{proposition}
\label{prop:OptCost}   
Let  $\Gamma= \langle  V, (w_u)_{u\in V}, \alpha,\beta \rangle$ be a \ourgame. 
We have that  $\optv(\Gamma) = \min \{\alpha, W\} (n-1)$.
\end{proposition}

\begin{proof}
Let $S\in \OPT(\Gamma)$, and let $G=G[S]$ with connected components
$G_1,...,G_r$, $V_i=V(G_i)$, $k_i = |V_i|$, 
and $W_i = w(V_i)$, for $1\leq i\leq r$.
Observe that the social cost of a disconnected graph can be expressed as the 
sum of the social cost of its connected components.  
Each connected component must be a tree of diameter at most $\beta$, 
otherwise a strategy profile  with smaller social cost could be 
obtained by replacing the connections on $V_i$ by such a tree. 
We can assume w.l.o.g. that, for $1\leq i\leq r$, the $i$-th connected 
component is a star graph $ST_{k_i}$ of $k_i$ vertexes.   
Since $C(ST_k) = \alpha (k-1)$ we have that
$$\displaystyle C(G)  =  \sum_{i=1}^r \alpha(k_i-1) + \sum_{i=1}^r k_i(W-W_i) 
=  \alpha(n-r)+  n W - \sum_{i=1}^r k_i W_i. $$
As $1 \leq k_i \leq  n -(r-1)$, we have  $W \leq \sum_{i=1}^r k_i W_i \leq (n-r+1)W$. 
Therefore,  $\alpha(n-r) + (r-1) W  \leq C(G)$. We consider two cases.

\smallskip
\noindent
\emph{Case 1:  $\alpha \geq W$}.  
We have   $W(n-1) \leq C(G)$. 
Since $C(I_n)= W(n-1) \leq C(G)$ 
and $G$ is an optimal graph, then $C(G)=W(n-1)$.

\smallskip
\noindent
\emph{Case 2: $\alpha < W$}. 
Now    $\alpha (n-1) \leq C(G)$.
As   $C(ST_n)= \alpha (n-1) \leq  C(G)$,   
the optimal graph $G$ has a social cost $C(G)=\alpha(n-1)$.
We conclude that    $\OPT=\min\{\alpha,W\}(n-1)$.
\end{proof}
Now we turn our attention to  the study of the \NE graph 
topologies showing that any \NE graph is either an independent 
set or a connected graph. 
\begin{proposition}
\label{prop:NEConnectedComponent}
Let  $\Gamma= \langle  V, (w_u)_{u\in V}, \alpha,\beta \rangle$ be a \ourgame. 
Every \NE graph of $\Gamma$ is either  connected or  the graph $I_n$, where $n=|V|$.
\end{proposition}

\begin{proof}
If $n\leq2$ the proposition follows immediately.
When  $n>2$, let us suppose that there  is a  \NE  
$S$ such that the graph $G=G[S]$ is 
not connected and different from $I_n$.  
In this case $G$ is composed of at least two 
different  connected components $G_1$ and $G_2$.  
Furthermore,  as $G\neq I_n$, we can assume that  $|V(G_1)| > 1$ 
as at least one of the connected components contains 
two vertexes connected by an edge.  
Let  $u\in V(G_1)$ be such that $S_u\neq \emptyset$. 
Let $x\in S_u$ and $v\in V(G_2)$. 
Let us consider the strategies $S_u'=S_u\setminus \{x\}$ and $S_v'=S_v \cup \{x\}$.
As $S$ is a \NE we know that $\Delta(S_{-u},S'_{u}) \geq 0$.
Let $G'= G[S_{-v},S_v']$, observe that  $d_{G'}(v,u)=2 \leq \beta$, 
therefore    $\Delta(S_{-v},S'_{v})  \leq  - \Delta(S_{-u},S'_{u}) -w_u <  0$.
This contradicts the hypothesis that $S$ is a \NE.
\end{proof}

Next we study the conditions under which  particular topologies are \NE graphs. 
Those results prove  that  celebrity games always have a \NE.
\begin{proposition}
\label{prop:NEStarIn}
Every \ourgame $\Gamma=\langle  V, (w_u)_{u\in V}, \alpha,\beta \rangle$ has a \NE. Furthermore,
when  $\alpha \geq w_{max}$, $I_n$ is a \NE graph, 
otherwise  $ST_n$  is a \NE graph but $I_n$ is not, where $n=|V|$.
\end{proposition}

\begin{proof}
When $\alpha \geq \wM$ let us show that $I_n$ is a \NE graph.  
Observe that $G=I_n$ is the outcome of a unique strategy profile $S$ in which 
$S_u=\emptyset$, for any $u\in V$. 
Let us consider a player $u$ and a strategy $S'_u\not = \emptyset$. 
The cost difference of player $u$ is then 
$\Delta(S_{-u},S_u')  = \alpha |S_u'| - \sum_{v\in S_u'} w_v= \sum_{v\in S_u'}(\alpha -w_v) \geq 0$. 
Therefore player $u$ has no incentive to deviate from  $S_u$ and $I_n$ is a \NE  graph.

When  $\alpha < \wM$, let $u$ be a vertex with $w_u=\wM$ and 
let  $ST_n$ be a star graph with $n$ vertexes in which the center 
is $u$, let us show that $ST_n$ is a \NE graph.
Consider the strategy profile $S$ in which $S_{u}=\emptyset$ and $S_v=\{u\}$, 
for any $v\in V$ different from $u$. 
Observe that the center $u$ is a vertex with maximum weight.   
As $\beta >1$  no player will get a cost decrease by connecting to more players. 
Furthermore, for $u\neq v$, $w_v+\alpha < w_v + \wM < W$. 
Thus $\alpha < W-w_v$  and  $v$ will not get any benefit by deleting the actual connection. 
The only remaining possibility is to reconnect to another vertex, 
but in such a case the cost cannot decrease. 
Therefore, $ST_n$ is a \NE graph.  
Notice that in this case $I_n$ can not be a \NE,  
as every player $u$ has incentive to connect with any other player $v$ such that $w_v=\wM$.
\end{proof}

To conclude  the study of \NE we characterize the \ourgames  
where  $I_n$ is the unique \NE graph.  
The negated condition characterizes those games  in which $ST_n$ is a \NE. 

\begin{proposition}
\label{pro:nash-disc}
Let  $\Gamma=\langle  V, (w_u)_{u\in V}, \alpha,\beta \rangle$ 
be a \ourgame on $n$ players with   $\alpha \geq \wM$.
If there is more than one vertex $u\in V$ with $\alpha >W-w_u$,  
then $I_n$ is the unique \NE graph  of $\Gamma$, 
otherwise $ST_n$ is a \NE graph of $\Gamma$.
\end{proposition}

\begin{proof}
Assume that,  for two vertexes $u\neq v$,  $\alpha > W-w_u$ and  $\alpha >W-w_v$, 
and that there exists a \NE graph  $G=G[S]$ different from $I_n$.  
By Proposition~\ref{prop:NEConnectedComponent}, $G$ is connected. 
Therefore, it has at least $n-1$ edges. 
Since, $\alpha > W-w_u$ and   $\alpha >W-w_v$, we have that $S_u = S_v = \emptyset$, 
otherwise $S$ would not be a \NE. 
Therefore, there must be a vertex, $z \neq u,v$ such that $|S_z| \geq 2$. 
Let $x,y\in S_z$ and consider the strategy  $S_z' = S_z \setminus \left\{x,y\right\}$. 
Then, $\Delta(S_{-z},S_z') \leq -2\alpha + W - w_z$. 
As $G$ is a \NE graph and we have that $2 \alpha > W-w_u + W-w_v$, we conclude  that
$W-w_z \geq 2\alpha >W-w_u+W-w_v$.  
Hence, $W<w_u+w_v-w_z < w_u+w_v$, which is impossible. 
In the case that there is at most one vertex $u$ with $\alpha > W-w_u$, 
 the strategy profile $S$, where  $S_u = \emptyset$, and  
$S_v = \left\{ u\right\}$, for all $v\neq u$,  is  a \NE. Furthermore $G[S]=ST_n$. 
\end{proof}
\begin{corollary}\label{cor:nash-disc}
Let  $\Gamma= \langle  V, (w_u)_{u\in V}, \alpha,\beta \rangle$ be a \ourgame on $n$ players. 
$I_n$ is the unique \NE graph  of $\Gamma$ if and only if 
$\alpha \geq \wM$ and there is more than one vertex $u\in V$ 
such that  $\alpha >W-w_u$.
\end{corollary}

Observe that in our model it is preferable to be an isolated node 
than to pay a huge amount for establishing a link. 
In fact, in a \NE graph either all nodes are isolated, or the graph is connected. 
Hence,  selecting an appropriate price per link   is a key fact  to guarantee the 
connectivity of the equilibrium graphs.

Finally, using this characterization, we can formally define  the subfamily of 
\ourgames that have always a connected \NE graph. 
Those games  have $ST_n$ as a \NE graph.
\begin{definition}
$\Gamma=\langle  V, (w_u)_{u\in V}, \alpha,\beta \rangle$ is 
a \stargame if $\Gamma$ has a \NE graph that  is connected.
\end{definition}
\begin{corollary}
\label{cor:NEStar}
For a  \ourgame 
$\Gamma=\langle  V, (w_u)_{u\in V}, \alpha,\beta \rangle$, 
the following  statements are equivalent.
\begin{itemize} 
\item $\Gamma$ is a  \stargame.
\item Either $\alpha < \wM$ or $\alpha \geq \wM$ and there is 
at most one $u\in V$ for which $\alpha > W-w_u$.
\item  $ST_n$ is a \NE graph of $\Gamma$.
\end{itemize}
\end{corollary}

Putting all together we can compute the  \PoS and, in some cases,  the \PoA.

\begin{theorem}
\label{theo:PoS} 
Let  $\Gamma$ be a  \ourgame. Then we have.
\begin{itemize} 
\item If $\Gamma$ is a \stargame, $\PoS(\Gamma)= 1$.
\item If $\Gamma$ is not a \stargame and   $\alpha\geq W$, then
$\PoS(\Gamma) = \PoA(\Gamma) = 1$.
\item If $\Gamma$ is not a \stargame and   $\alpha < W$, then
$\PoS(\Gamma) =   \PoA(\Gamma) = W/\alpha >1$. 
\end{itemize}
\end{theorem}
\begin{proof}

From Proposition~\ref{prop:OptCost}, we have that  
$\optv(\Gamma) = W(n-1)$ if $\alpha\geq W$ and  
$\optv(\Gamma) = \alpha (n-1)$, otherwise. 
When $\Gamma$ is a \stargame, by Corollary ~\ref{cor:NEStar} 
we know that  $ST_n$ is  a \NE graph. Let us see that in \stargames 
it can only occur that $\alpha < W$.
If $\alpha < \wM$, clearly  $\alpha < W$. If $\alpha \geq \wM$,  
by Corollary~\ref{cor:NEStar} there is at most one $u\in V$ for which $\alpha > W-w_u$. 
Assuming that  $w_{u_1}\leq \ldots \leq w_{u_{n-1}} \leq w_{u_n}$, we have that 
$W>W-w_{u_1} \geq \ldots \geq W-w_{u_{n-1}} \geq W-w_{u_n}$, 
and then $W-w_{u_{n-1}} \geq \alpha$.
Hence, $\PoS(\Gamma)= 1$.

When $\Gamma$ is not a \stargame,  $I_n$ is the unique \NE graph. 
Thus, when $\alpha \geq W$ we have, 
$\PoS(\Gamma) = \PoA(\Gamma) = 1$ and,   
when $\alpha < W$ we have, $\PoS(\Gamma) = \PoA(\Gamma) =  W/\alpha >1$.
\end{proof}

\section{Critical distance and diameter in  Nash equilibrium graphs}
\label{sec:ne-diameter}

In this section we analyze the diameter of \NE graphs and  
its relationship with the parameters defining the game.  
We are interested only in games in which  \NE graphs with 
finite diameter exist, thus we only consider   \ourgamens. 
In stating the characterization, nodes with a high weight with 
respect to the link cost play a fundamental role and it is worth to give them a name.  

\begin{definition}
Let $\Gamma=\langle  V, (w_u)_{u\in V}, \alpha,\beta \rangle$ be a \ourgame.
We say that a vertex $u \in V$  is a  \emph{celebrity} if  $\alpha <w_u$.
\end{definition}

Given a  celebrity $u$, any other node $v$ with $d(u,v) > \beta$  
has an incentive to pay for connecting to $u$. 
Thus, in any  \NE graph $G$,  every celebrity node $u$  
satisfies that $diam(u)\leq \beta$. 

In some of the proofs of the following results we refer to a   
set of \emph{critical nodes} $z\in V$ of a graph $G=(V,E)$  
with respect to a  node $u$ and an edge $\{x,y\}$. 
Critical is used  in the sense that as all the shortest paths from $u$ to $z$ use $\{x,y\}$, 
removing the edge $\{x,y\}$ results in an increase of the distance from $u$ to $z$.  
We use the notation
$$A^G_{\{x,y\}}(u)=\{z \in V | \text{ all the shortest paths in $G$ from $u$ to $z$ use the edge }  \{x,y\}\}$$  
We drop the explicit reference to $G$ whenever $G$ 
is clear from the context.

\begin{proposition}\label{DiameterGeneral}
Let $\Gamma= \langle  V, (w_u)_{u\in V}, \alpha,\beta \rangle$
be a \stargame. If $G$ is  a \NE graph of $\Gamma$, then  $diam(G)\leq 2\beta + 1$.
\end{proposition} 

\begin{proof}
Let $S$ be a \NE  of $\Gamma$ such that  $G=G[S]$. Assume that $diam(G) \geq  2\beta + 2$. 
Then, there are two nodes $u,v\in V$ such that $d(u,v) = 2\beta + 2$. 
Consider a shortest path from $u$ to $v$,  $u=u_0,u_1,\dots, u_{2\beta +1}, u_{2\beta +2}= v$. 

Let $A_u=\{x \in V | d(u,x) \leq \beta\}$ and let $A_{u_1}=\{x \in V | d(u_1,x) \leq \beta \}$. 
Let us show  that if a node  $x\in A_u\cup A_{u_1}$, then $d(x,v)>\beta$. If $x\in A_u$ then $d(x,v) > \beta$, otherwise $d(u,v)\leq d(u,x) + d(x,v) \leq 2 \beta$ contradicting the fact that $d(u,v) = 2\beta + 2$.
Moreover, if  $x\in A_{u_1}$ then $d(x,v) > \beta$, otherwise $d(u,v)\leq 1+ d(u_1,x) + d(x,v) \leq 2 \beta +1$ which also contradicts the fact that $d(u,v) = 2\beta + 2$.

Consider the edge $\{u,u_1\}$. Then, either $u_1\in S_u$ or $u\in S_{u_1}$.
In the case that $u_1\in S_u$, let $S_u'= S_u\setminus\{u_1\}$ and $S_v'= S_v\cup \{u_1\}$. 
Observe that,
$$ \Delta(S_{-u}, S_u') \leq -\alpha + w(A_{\{u,u_1\}}(u) \cap A_u)$$ 
By the previous remark about distances, we know that all the vertexes  
$x\in A_{\{u,u_1\}}(u) \cap A_u$ verify $d(x,v)>\beta$, 
but after adding $\{v,u_1\}$ all of them and $u$ become at distance $\leq \beta$ from $v$, 
therefore
$$\Delta(S_{-v}, S_v') \leq \alpha - w_u- w(A_{\{u,u_1\}}(u) \cap A_u).$$
Hence,   $\Delta(S_{-u}, S_u')+ \Delta(S_{-v}, S_v') \leq -w_u < 0$. 
Therefore, either  $\Delta(S_{-u}, S_u')<0$ or $\Delta(S_{-v}, S_v')<0$ and then $S$ can not be a \NE.

The case $u\in S_{u_1}$, follows in a similar way by 
interchanging the roles of $u$ and $u_1$. 
\end{proof}

The previous result can be refined to get better bounds on the 
diameter  when all the nodes are celebrities or when  at least one 
of the nodes is a celebrity. 

\begin{property}\label{DiameterNE} 
Let $\Gamma= \langle  V, (w_u)_{u\in V}, \alpha,\beta \rangle$ be a 
\ourgamen and let  $G$ be a  \NE graph of  $\Gamma$,  then
\begin{itemize}
\item   if $\wm\leq \alpha  < \wM$,   $diam(G)\leq 2 \beta$ and, 
\item  if  $\alpha < \wm$, $diam(G)\leq \beta$.
\end{itemize}
\end{property}
\begin{proof}
When $\wm\leq \alpha  < \wM$, there is a celebrity $u\in V$ with  $w_u > \alpha$. 
We know that $diam(u) \leq \beta$. 
Let $x$ and $z$ be any two different nodes of $G$, then $d(x,u)\leq \beta$ and 
$d(z,u)\leq \beta$. Therefore, $d(x,z)\leq d(x,u) + d(z,u)\leq 2\beta$ and the claim follows.
When  $\alpha < \wm$, each $u\in V$ is a celebrity, thus  $diam(u)\leq \beta$. 
Therefore  $diam(G)\leq \beta$.   
\end{proof}

For \NE trees we have a trivial lower bound of $2$ on  the diameter as a star is a \NE graph. 
For non-tree \NE graphs we  provide a  lower bound on the diameter. 
We first prove  a technical result.

\begin{lemma}\label{lem:cycle}
Let $\Gamma= \langle  V, (w_u)_{u\in V}, \alpha,\beta \rangle$ be a \stargame. 
In a \NE graph of $\Gamma$ containing at least one cycle, if $u$  
is  a node of a cycle and  $diam(u) \leq \beta - k$, for some $k\geq 1$, 
then the length of any cycle containing $u$ is bigger than $2k+2$.
\end{lemma}
\begin{proof} 
Let us suppose that $S$ is a \NE and that $G=G[S]$  contains a 
cycle $C$ through  a node $u$ such that $diam(u) \leq \beta - k$, 
for some  $k\geq 1$. Assume that $C$ is the shortest cycle containing 
$u$ and that the length $\ell$ of  $C$  verifies   $\ell \leq 2k+2$.  
We split the proof in two cases, depending on the parity of  $\ell$.

\smallskip
\noindent
\emph{Case 1: $C$ has odd length, $\ell = 2i+1$}.  
Let $v_1,v_2$ be the two vertexes in $C$ that are at  distance  
$i$ of $u$ in $C$, as $C$ is of minimal length  $d(u,v_1)=d(u,v_2)= i$.  
By our hypothesis,  $2i+1\leq 2k+2$ and thus $i \leq k$. 
Assume w.l.o.g. that $v_2\in S_{v_1}$ and consider the strategy 
$S_{v_1}'=S_{v_1}\setminus \{v_2\}$. Let $G'=G[S_{-v_1},S_{v_1}']$. 
Notice that    $d_{G'}(v_2, u) = i$. 
Therefore, $diam_{G'}(v_1)\leq k+\beta-k = \beta$, 
by selecting a path  going through $u$, 
so $\Delta(S_{-v_1},S_{v_1}') \leq -\alpha < 0$ and $G$ would not be a \NE graph.

\smallskip
\noindent
\emph{Case 2: $C$ has even length, $\ell = 2i$}. 
Let $v$ be the antipodal vertex to $u$, at distance $i$ 
from $u$ in $C$ and let  $v_1,v_2$ be the two vertexes 
in $C$ that are at  distance  $i-1$ of $u$ in $C$. 
By our hypothesis,  $2i\leq 2k+2$ and thus $i -1\leq k$. If $v\in S_{v_1}$, 
consider the strategy   $S_{v_1}'=S_{v_1}\setminus \{v\}$. 
Using the same arguments as in Case 1 and the fact that the distance 
from $v_1$ to $u$ in $C$ is $\leq k$, we  conclude that $S$ is not a \NE. 
The same happens when $v\in S_{v_2}$. It remains to consider the case 
in which $v_1,v_2\in S_v$.  Consider the strategy $S_v'=(S_v \cup \{u\})\setminus\{v_1,v_2\}$.
Now all shortest paths in $G$ from $v$ passing through $v_1$ or $v_2$ can be 
rerouted trough $u$ with and increment in length of at most $i-1\leq k$.  
Therefore,  $diam_{G'}(v) \leq 1 + \beta -k \leq \beta$.  
Thus $\Delta(S_{-v},S_{v}')\leq -\alpha < 0$ and $G$ would not be a \NE graph.
\end{proof}

\begin{proposition}\label{DiameterLowerBound}
Let $\Gamma= \langle  V, (w_u)_{u\in V}, \alpha,\beta \rangle$
be a \stargame and let $G$  be a \NE graph of $\Gamma$. 
If $G$ is not a tree, then  $diam(G)> \beta/2$.
\end{proposition}

\begin{proof} 
Let $G$ be a \NE graph containing at least one cycle. If $diam(G)\geq  \beta$, the claim holds.
Assume that  $diam(G)\leq  \beta-1$. 
We know that the length of the shortest cycle $C$ is  $\leq 2\, diam(G)+1$. 
Let $u$ be any node of $C$. Then, we have  $diam(u)\leq diam(G) = \beta - (\beta-diam(G))$. 
By Lemma~\ref{lem:cycle}, $2\, diam(G)+1 > 2(\beta-diam(G))+2$. 
The last inequality implies $diam(G) > \beta/2$.
\end{proof}


\section{MaxBD network creation games versus celebrity games}
\label{sec:MaxBD}

In this section we show that  \maxbds are equivalent to  \ourgames
where all players are celebrities. 
Let us formalize the definition of \maxbd taken from \cite{BiloGP:15}. 

A \emph{\maxbd} $\Gamma$ 
is defined by a tuple $\langle V, D \rangle$ where
$V=\{1,\dots, n\}$ is the  set of players and 
$D$, $1\leq D\leq n-1$,  is an integer representing the bound 
on the diameter of each node $v\in V$. 
Concepts like  strategy of a player, strategy profile, and outcome 
graph are defined  as in  the  \ourgame model.
The \emph{cost} of player $u$ in the strategy profile  
$S$ is $\cmaxbd_u(S) = |S_u|$, if $diam_{G[S]}(u)\leq D$; $\cmaxbd_u(S) = + \infty$, otherwise.
The social cost of $S$ is  $\Cmaxbd(S)=\sum_{u\in V} \cmaxbd_u(S)$.
Notice that by the definition of \maxbd,  any strategy profile $S$ that is either 
a  \NE or a social \OPT satisfies $diam_{G[S]}(u)\leq D$ and, therefore $C(S)=\alpha \, \Cmaxbd(S)$.

In the following we show how a \maxbd  can be translated, preserving \NE, 
to different  instances of \ourgames.  A \maxbd  can be seen as a \ourgame 
in which the weights of each one of the players are  large enough so that buying 
a link is more suitable  than having  an eccentricity greater than the given 
distance bound. On the other hand, we show that every  
\ourgame with $\alpha < w_{min}$  corresponds to a \maxbd,  again preserving \NE.

\begin{proposition}\label{CelebritiesVSMaxBD}
Let $V$ be a set of players and $\beta>1$. 
Let $\Gamma=\langle V, \beta\rangle$ be a \maxbd  and   
$\Gamma'=\langle V, (w_v)_{v\in V}, \alpha, \beta \rangle$   
be a \ourgame where $\alpha < w_{min }$.     
Then, $\NE(\Gamma)=\NE(\Gamma')$.
\end{proposition}

\begin{proof}
Let us prove first that  $\NE(\Gamma)\supseteq \NE(\Gamma')$. 
Assume that $S \in \NE(\Gamma')$. 
By Property~\ref{DiameterNE}, $diam_{G[S]}(u) \leq \beta$ and 
this implies that $c_u(S)=\alpha |S_u|$.
Let us suppose that  $S$ is not a \NE for $\Gamma$. 
Then there exists a player $u \in V$ and a strategy 
$S_u'$  such that $\cmaxbd_u(S')< \cmaxbd_u(S)=  |S_u|$, where $S'=(S_{-u},S_u')$. 
Hence, the only possibility is that $diam_{G[S']}(u) \leq \beta$ and $|S_u'| < |S_u|$. 
Therefore $c_u(S')< c_u(S)$ contradicting the fact that $S\in \NE(\Gamma')$.

It remains to  show that  $\NE(\Gamma) \subseteq \NE(\Gamma')$. 
Let  $S\in\NE(\Gamma)$. We know that   $diam(G[S]) \leq \beta$ and
$\cmaxbd_u(S)=|S_u|$, for  $u\in V$. For $\Gamma'$, we have that $c_u(S)= \alpha |S_u|$. 
Now let us assume that $S$ is not a \NE of $\Gamma'$.
Then, there exists 
$u\in V$ and a strategy $S_u'$ such that $c_u(S)>c_u(S')$, 
where $S'=(S_{-u},S_u')$. Since $w_v > \alpha$, 
then we have that $c_u(S') = \alpha |S_u'|  + \sum_{\{v \mid d_{G[(S')]}(u,v)> \beta \}} w_v \geq \alpha (|S_u'|  + |\{v \mid d_{G[S']}(u,v)>\beta \}|)$.   
Consider the strategy profile $S''=(S_{-u},S_u'')$, 
where  $S_u''=S_u' \cup \{v \mid d_{G[S']}(u,v)>\beta \}$. 
We have    $diam_{G[(S'']}(u)\leq \beta$. 
Thus,  $c_u(S'')=\alpha |S_u''|$. 
Combining the inequalities 
$c_u(S) = \alpha |S_u| > c_u(S'')=\alpha |S_u''|$.  
Then, $|S_u|>|S_u''|$ contradicting the fact that $S \in \NE(\Gamma)$.
 \end{proof}

The previous correspondences allow us to get a relationship on the \PoA and the \PoS,
\begin{corollary}
Let $V$ be a set of players and $\beta>1$. 
Let $\Gamma=\langle V, \beta\rangle$ be a \maxbd and   let $\Gamma'=\langle V, (w_v)_{v\in V}, \alpha, \beta \rangle$   be a \ourgame where  $\alpha < w_{min}$.     Then,
\begin{itemize}
\item $\PoS(\Gamma)=\PoS(\Gamma')=1$,
\item $\PoA(\Gamma)=\PoA(\Gamma')$.
\end{itemize}
\end{corollary}

\begin{proof}
We know  by Proposition~\ref{prop:NEStarIn} that  the  star tree is a social optimum as well 
as a \NE for \ourgames when $\alpha < w_{min}$. 
The same occurs for \maxbds as it was shown in Theorem 3.3 of  \cite{BiloGP:12}. 
Hence, $\PoS(\Gamma)=\PoS(\Gamma')=1$.

For the \ourgame $\Gamma'$, we have that 
$$\PoA(\Gamma)= \frac{ \alpha  \max_{S \in  \NES(\Gamma)} \{ |E(G[S])|\}} {\alpha (n-1)} =
\frac{ \max_{S \in  \NES(\Gamma)}   \{|E(G[S])|\}}  {(n-1)}.$$
By Proposition \ref{CelebritiesVSMaxBD},
$\NE(\Gamma)=\NE(\Gamma')$. Thus \NE of $\Gamma'$ have 
diameter $\leq \beta$ and then  we can conclude that $\PoA(\Gamma)=\PoA(\Gamma')$. 
\end{proof}

Hence, the  upper bound on the \PoA of \maxbds shown 
in \cite{BiloGP:15} is also an upper bound  for \ourgames. 
In the subsequent sections we consider the general case where the 
assumption  $\alpha < w_{min}$ is not required.

We have considered here only the uniform version of the \maxbds in 
which the eccentricity bound is equal for all the nodes. \cite{BiloGP:15} 
considers also a non uniform version in which each node has a different 
eccentricity requirement.  It is easy to extend  Proposition~\ref{CelebritiesVSMaxBD} 
to show that the set of \NE is preserved provided that the eccentricity bounds are the 
same in both games and $\alpha < w_{min}$.  
Therefore, non-uniform \ourgames have unbounded \PoA,  
as it was shown for the non-uniform \maxbds in \cite{BiloGP:15}. 


%
\section{Bounding the price of anarchy}
\label{sec:poa-lower-bound}
We provide here  bounds on the contribution of the edges and the weights to the social 
cost of \NE graphs. 
Those bounds allow us to provide a bound on the \PoA. 
Our next result establishes an upper bound on  the 
\PoA in terms of $W$  and $\alpha$. 

\begin{lemma}
\label{lem:ub-poa}
For a \ourgamen  $\Gamma= \langle  V, (w_u)_{u\in V}, \alpha,\beta \rangle$, 
$\PoA (\Gamma)\leq W/\alpha$. 
\end{lemma}
\begin{proof}
Let $S$ be a \NE of $\Gamma$ and let $G=G[S]=(V,E)$.  
As $S$ is a \NE, no player has an incentive to deviate from $S$. 
Thus,  for any $u\in V$,
\[0\leq \Delta(S_{-u},\emptyset) \leq - \alpha|S_u | + w(\{v\mid d(u,v) \leq \beta\})  - w_u.\]
Summing up, for all $u\in V$, we have
\[
\displaystyle 0\leq \sum_{u\in V} (- \alpha|S_u | + 
\sum_{\{v\mid d(u,v) \leq \beta\} } w_v  - w_u) = 
-\alpha |E| + \sum_{u \in V}\sum_{\{v\mid d(u,v) \leq \beta\} } w_v  -  W.
\]
Therefore,
\begin{align*}
C(G)  & =  \alpha |E| + \sum_{u \in V} \sum_{\{v\mid d(u,v) > \beta\} } w_v \\
          & \displaystyle\leq \sum_{u \in V} \left(\sum_{\{v\mid d(u,v) \leq \beta\} } w_v  + 
          \sum_{\{v\mid d(u,v) > \beta\} } w_v \right)-W 
           =  (n-1) W.
\end{align*}

Hence, $\PoA (\Gamma)\leq \frac{ (n-1)W}{\alpha (n-1)} =\frac{ W}{\alpha}$.
\end{proof}

Using the previous lemma we can get an $O(n)$ upper bound on the \PoA  of 
\ourgamens. Let us see that this upper bound can be improved by bounding 
the weight component and the link component of the social cost, separately.

Define the \emph{weight component} of the social cost, for a critical distance 
$\beta$, $W(G,\beta)$, as  
\[
W(G,\beta)= \sum_{u \in V(G)} \sum_{\{v\mid d(u,v)>\beta\}} w_v = 
\sum_{\{\{u,v\}\mid d(u,v)>\beta \}} (w_u +w_v).
\]

\begin{lemma} 
\label{prop:w-cost}
Let $\Gamma= \langle  V, (w_u)_{u\in V}, \alpha,\beta \rangle$ be a \ourgamen.
In a \NE graph $G$,  $W(G,\beta) = O(\alpha n^2/\beta)$.
\end{lemma}

\begin{proof}
Let $S$ be a \NE and  $G=G[S]$ be  a \NE graph.  Let  $u \in V$ and   let $b=diam(u)$. 
Recall that, by Proposition~\ref{DiameterGeneral}, $b\leq 2\beta +1$.  
We have three cases.

\bigskip

\noindent
\emph{Case 1: $b < \beta$}.  
For any node  $v\in V\setminus\{ u\}$ consider the strategy 
$S'_v=S_v\cup \{u\}$, and let $G'=G[S_{-v},S'_v ]$. 
By connecting to $u$ we have $diam_{G'}(v)\leq \beta$ and, 
as  $S$  is a \NE, we have  
\[
\Delta(S_{-v},S'_v )= \alpha - \sum_{\{x\mid d_G(x,v)>\beta\}} w_x\geq 0.
\]
Therefore we have
\[\sum_{\{x\mid d_G(x,v)>\beta\}} w_x  \leq \alpha.\] 
As $b<\beta$ we conclude that 
\[W(G,\beta)     \leq  n \alpha.\] 
Since  $1 <\beta \leq n-1$, we get  $n/\beta \leq \alpha n^2/\beta$. 

\bigskip

\noindent
\emph{Case 2:  $ b \geq \beta$ and $b \geq 6$}.
For $1\leq i\leq b$, consider the set $A_i(u) = \{v\mid d(u,v)=i\}$ and the sets
\begin{align*}
C_1 &=\{v\in V \mid 1\leq d(u,v)\leq b/3\} = \cup_{1 \leq i\leq b/3} A_i(u),\\
C_2 &=\{v\in V \mid b/3 < d(u,v)\leq 2b/3\} =  \cup_{b/3 < j\leq2 b/3} A_j(u),\\
C_3 &=\{v\in V \mid 2b/3 < d(u,v)\leq b\} = \cup_{2b/3 < k \leq  b} A_k(u).
\end{align*}
As $b=diam(u)$,  $A_\ell(u)\neq \emptyset$, $1\leq \ell\leq b$, 
and all those sets constitute a partition of $V\setminus\{u\}$. 
As $ b\geq 6$,  for each $\ell$,  $1\leq \ell \leq 3$,  $C_\ell$  contains 
vertexes at a $b/3\geq 2$ different distances. 
Therefore, for  $1\leq \ell\leq 3$,  it must exist $i_\ell$ such that 
$A_{i_\ell}(u)\subseteq C_\ell$ and $|A_{i_\ell}(u)| \leq 3n/b$, 
otherwise the total number of elements in $C_\ell$ would be bigger than $n$.    

For any $v\in V$, let $S'_v=  (S_v \cup A_{i_1}(u) \cup A_{i_2}(u)\cup A_{i_3}(u))\setminus\{v\}$ 
and let  $G'=G[S_{-v},S'_v]$. Since $b\leq 2 \beta +1$,
 we have that $b/3 < \beta$. Hence, by construction,  $diam_{G'}(v)\leq \beta$. 
 Therefore, as $S$ is a \NE, we have
\[
0 \leq \Delta(S_{-v},S'_v )\leq \frac{9 n \alpha}{\beta }- \sum_{\{x\mid d_G(x,v)>\beta\}} w_x.
\]
Thus,
\[
\sum_{\{x\mid d_G(x,v)>\beta \}} w_x \leq \frac{9 n \alpha}{\beta}\text{ and }
W(G,\beta)\leq \frac{9 n^2 \alpha}{\beta }.
\]

\bigskip

\noindent
\emph{Case 3: $b \geq \beta$ and $b\leq 6$}. 
Consider the sets $A_i(u) = \{v\mid d(u,v)=i\}$, $0\leq i\leq b$, 
and the sets $C_0=\{v\in V \mid d(u,v) \text{ is even}\}$ and $C_1=V\setminus C_0$. 
Both sets are non-empty and one of them must have $\leq n/2$ vertexes. 
By connecting to all the vertexes in the smaller of those sets the diameter 
of the resulting graph is 2. Therefore, using a similar argument as in case 2, 
we get
\[
W(G,\beta)\leq \frac{ n^2 \alpha}{2},  
\]
which is $O(n^2/\beta)$ as  $\beta<6$.
Which concludes the proof.
\end{proof}

Our next result provides a bound for the number of edges in a \NE graph.

\begin{lemma} 
\label{prop:link-cost}
Let $\Gamma= \langle  V, (w_u)_{u\in V}, \alpha,\beta \rangle$ be a \ourgamen.
In a \NE graph $G$,  $|E(G)| \leq  n-1 +\frac{3n^2}{\beta}$.
\end{lemma}

\begin{proof}
Let $S$ be a \NE of $\Gamma$ and let $G=G[S]=(V,E)$. 
Let $u$ be a node in $V$.
For any $v \in S_u$, recall that $A_{\{u,v\}}(u)$ denotes the set of nodes $z$ 
such that all shortest paths from $u$ to $z$ use the edge $\{u,v\}$. 
Observe that  $v\in A_{\{u,v\}}(u)$
and that, for $v,v'\in S_u$ with $v\neq v'$, $A_{\{u,v\}}(u)\cap A_{\{u,v'\}}(u) = \emptyset$.

Let   $B(G)$  be the set of  bridges of  $G$, recall that $|B(G)|\leq n-1$.  
For $u\in V$, let $\overline{B}(u)= \{x\in S_u \mid  \{u,x\}\notin B(G)\}$. 
Observe that $|E| = |B(G)|  + \sum_{u\in V} |\overline{B}(u)|$.  

Let us show that for any $v \in S_u$ such that  $\{u,v\}$ is not a bridge,  
there exists $z\in A_{\{u,v\}}(u)$ such that $d(u,z) > \beta/3$. 

Let us suppose that  $\{u,v\}$ is not a bridge and that,  
for every  $z\in A_{\{u,v\}}(u)$   $d(u,z) \leq \beta/3$.  
In such a case there must be some edge $\{x,y\}$ with  
$x \notin A_{\{u,v\}}(u)$ and $y \in A_{\{u,v\}}(u)$. 
Furthermore, we can select $x$ so that $x\neq u$ and such 
that there is a shortest path $P$ from $u$ to $x$ using only 
vertexes in  $V \setminus A_{\{u,v\}}(u)$.  
Observe that $d(u,x) \leq d(u,y) +1$.  Furthermore,  for  $z \in  A_{\{u,v\}}(u)$,  
there exists  a path from $u$ to $z$ that follows $P$  from $u$ to $x$, 
the edge  $\{x,y\}$, a shortest path from $y$ to $v$ 
(part of a shortest path to $u$ through $A_{\{u,v\}}(u)$), 
and a shortest path  from $v$ to $z$  (through $A_{\{u,v\}}(u)$). 
Notice that  $d(u,x) \leq \beta/3 +1$, $d(y,v)\leq \beta/3 -1$,  and $d(v,z)\leq \beta/3 -1$. 
Hence,  there is a path from $u$ to $z$ of 
distance $\leq (\beta/3+1)+1+(\beta/3-1)+(\beta/3-1) = \beta$ which does not use $\{u,v\}$. 
Thus,  $u$  has incentive to remove  $\{u,v\}$ since 
$\Delta (S_{-u}, S_u \setminus \{v\})= -\alpha < 0$, which contradicts the fact that $S$ is a \NE. 

Therefore, for $v\in \overline{B}(u)$,  there exists $z\in A_{\{u,v\}}(u)$ 
such that $d(u,z) > \beta/3$  and as all the predecessors of $z$ in a 
shortest path from $u$ belong to  $A_{\{u,v\}}(u)$, we have $|A_{\{u,v\}}(u)| > \beta/3$. 
Observe that $n \geq  \sum_{\{v \in S_u | v \in \overline{B}(u)\}} |A_{{u,v}}(u)| \geq |\overline{B}(u)| (\beta/3)$, thus 
$|\overline{B}(u)| \leq \frac{3 n}{\beta}$. 
Finally, combining the two bounds, we have
$|E| = |B(G)|  + \sum_{u\in V} |\overline{B}(u)| \leq (n-1) + \frac{3 n^2}{\beta}$.
\remove{
Also, it is clear that if $v'\in A_v(u)$ with $u-v-v_1-v_2-...-v_k-v'$ a minimal length 
path from $u$ to $v$ then $v_i \in A_{v}(u)$ for $1 \leq i \leq k$ (and also $v\in A_v(u)$ trivially). 
This gives us that the number of pairs $(u,v')$ 
such that there exists $v$ with $v' \in A_v(u)$ is at least $(|E|-(n-1))\beta/3$. 
On the other hand, there are at most $\binom{n}{2}$ of the same pairs. 
Putting this two inequalities together we get: 
$$|E| < \frac{3}{2\beta}n^2-\frac{3}{2\beta}n+n-1 = O(n^2/\beta)$$
}
\end{proof}

Observe that, the previous results  jointly with  $\optv(\Gamma) =\alpha(n-1)$,  
leads us to the following  upper bound of the \PoA. 

\begin{theorem} 
For a \ourgamen  $\Gamma$,  $\PoA(\Gamma) = O(\min\{n/\beta,W/\alpha\})$.
\end{theorem} 

We finalize this section showing a family of \ourgamens having $\PoA=\Omega(n)$, for $\beta=2$.

\begin{figure}
\begin{center}
\begin{tikzpicture}[every node/.style={circle,scale=0.8}, node distance=7mm, >=latex]
\node[draw](a) at (0,5) {$\alpha$};
\node[draw](b) at (2,4) {$w_1$};
\node[draw](c) at (4,4) {$w_2$};
\node[draw](d) at (6,5) {$\alpha$};

\node[draw](a1) at (0,1) {$\alpha$};
\node[draw](b1) at (2,2) {$w_3$};
\node[draw](c1) at (4,2) {$w_4$};
\node[draw](d1) at (6,1) {$\alpha$};

\draw[->] (b) to node {}(a);
\draw[->] (c) to node {}(b);
\draw[->] (c) to node {}(d);
\draw[->] (b1) to node {}(a1);
\draw[->] (b1) to node {}(c1);
\draw[->] (c1) to node {}(d1);
\draw[->] (b) to node {}(b1);
\draw[->] (b) to node {}(c1);
\draw[->] (c1) to node {}(c);
\draw[->] (c) to node {}(b1);
\end{tikzpicture}
\end{center}
\caption{A \NE for the game $\Gamma(4,\alpha,w)$. \label{fig:LB}}
\end{figure}
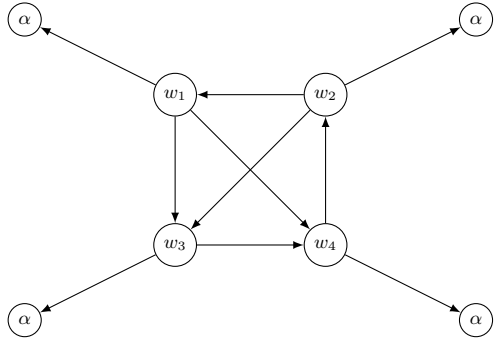

\begin{lemma}
Let  $k>2$,  $\alpha>0$ and let $w=(w_1,\dots w_k)$ be a positive weight assignment.  
There is a  \ourgamen  $\Gamma=\Gamma(k,\alpha,w)$ with $n=2k$ players 
and $\beta=2$  having $\PoA(\Gamma) > \frac{3 n}{8}$.   
\end{lemma}
\begin{proof}
Consider the game $\Gamma_k= \langle  V, (w_u)_{u\in V}, \alpha,\beta \rangle$, where 
\begin{itemize}
\item $V=\{u_1,\dots, u_k\}\cup \{v_1\dots,v_k\}$,
\item $w(u_i)=\alpha$ and $w(v_i)=w_i$, for $1\leq i\leq n$, 
\item $\beta =2$.
\end{itemize}

Consider any strategy profile $S$ where, 
for $1\leq i\leq k$,  $\{u_1,\dots, u_k\}\cap S_{v_i}=\{u_i\}$ and  $S_{u_i}=\emptyset$,  
and such that in $G[S]$ the subgraph induced by $\{v_1\dots,v_k\}$ is a clique.  
An example of such a strategy, for $k=4$,  is given in Figure~\ref{fig:LB}.

Observe that there is no vertex in $G[S]$ that is at distance 1 of more than 
one vertex in $\{u_1,\dots, u_k\}$. Furthermore, any edge $(u,v)$ lies in the 
unique shortest path from $u$ to a vertex in $\{u_1,\dots, u_k\}$. Therefore  $S$ is a \NE.

We have $C(G[S]) = \alpha \left(\frac{k(k-1)}{2} + k\right) + \alpha k(k-1) = \alpha (3 k(k-1) + 2k)/2$. 
As a star tree is an \OPT graph and $n=2k$, we conclude that
$$\PoA(\Gamma)=   \frac{\alpha \frac{ \frac{3n}{2}(\frac{n}{2}-1)  + 2 \frac{n}{2}}{2}}{\alpha (n-1)}
                      = \frac{3n}{8}\frac{(n - 1) +\frac{1}{3}}{(n-1)} =  \frac{3n}{8} \left(1 +\frac{1}{3(n-1)}\right).$$
\end{proof}

\section{Price of anarchy on Nash equilibrium trees}
\label{sec:poa-trees}
Now we complement the results of  the previous sections by providing a 
constant upper bound on the \PoA  when we restrict the \NE graphs to be trees.
We can find in the literature different  models for which the diameter or the \PoA 
can not be proved to be constant on  general \NE graphs, but  they are shown 
to be  constant in the  case of  \NE trees (see for example \cite{AlonDHL:13,AlonDHK:14,Ehsani:2015}).

In order to get a tighter upper bound for the \PoA on \NE trees, we first improve the 
bound on the diameter of \NE trees to   $\beta + 1$. 

\begin{proposition}
\label{pro:diamtrees}
Let $\Gamma= \langle  V, (w_u)_{u\in V}, \alpha,\beta \rangle$
be a \ourgamen. If $T$ is  a \NE tree of $\Gamma$,  $diam(T)\leq \beta + 1$.
\end{proposition}

\begin{proof}
Let  $T$ be a tree such that $T=G[S]$ where $S$ is a \NE of $\Gamma$. 
Let $d=diam(T)$ and let $P: u=u_0, u_1, \dots, u_d$ be a diametral path of $T$. 
Assume that $d> \beta +1$.    For $1\leq i < d$, let  $T_i$   
be  the connected subtree containing $u_i$ after removing 
edges $\{u_{i-1}, u_i\}$ and $\{u_i,u_{i+1}\}$.   
As $P$ is a diametral path, both   $u$ and $u_d$ are leaves in $T$. 
Furthermore, $T_1$ and $T_{d-1}$ are star trees.   
In general, the distance from the leaves of any $T_i$ to both $u$ and $u_d$ is at most $d$.  

We  consider two cases depending on who is paying for the connections to the end points of $P$.  

\noindent
\emph{Case 1: $u \in S_{u_1}$ or $u_d\in S_{u_{d-1}}$}.  
W.l.o.g. assume that $u_d\in S_{u_{d-1}}$.
As $S$ is a \NE we have  $w_{u_d} \geq \alpha$.  Consider the strategy
$S'_{u_1}=S_{u_1} \cup \{u_{d-1}\}$, then $\Delta(S_{-u_1},S'_{u_1})\leq \alpha - w_{u_d}- w_{u_{d-1}} < 0$ and $T$ can not be a \NE graph.

\noindent
\emph{{Case 2:} $u_1\in S_u$ and $u_{d-1}\in S_{u_d}$}. 
When $\beta \geq 3$. Set
$S'_u=S_u-\{u_1\}\cup \{u_2\}$ and $T'=G[(S_{-u},S'_u)]$. 
Observe that, for $x\in T_1$, $d_{T'}(u,x)\leq 3 \leq \beta$ and, 
for $x\notin T_1\cup\{u\}$, $d_{T'}(u,x)=d_T(u,x)-1$.  
Therefore, $\Delta(S_{-u},S'_u) \leq -w_{u_{\beta+1}}< 0$.  
Therefore, $T$ is not a \NE graph. 

The previous argument fails when  $\beta = 2$ as  there might be $x\in T_1$ with $d_{T'}(u,x)= 3$. 
From  Proposition~\ref{DiameterGeneral}, we know that $d\leq 2\beta+1 \leq 5$. 
Let us see that it can not be the case that $d=4$ or  $d=5$.  
Let $S'_u=S_u-\{u_1\}\cup \{u_{d-1}\}$ and 
$S'_{u_{d}}=S_{u_{d}}-\{u_{d-1}\}\cup \{u_{1}\}$. 
Let $T^1=G[(S_{-u},S'_u)]$ and $T^2=G[(S_{-u_d},S'_{u_d})]$.

When $d=4$,  for any $x\in T_2$,  $d_{T^1}(u,x) = d_{T}(u,x)$   
and $d_{T^2}(u_4,x) = d_{T}(u_4,x)$.  Therefore, we have 
\[\Delta(S_{-u},S'_u)=  w(T_1) -  w(T_{3}) - w_{u_4} \text{ and } 
\Delta(S_{-u_4},S'_{u_4})= w(T_{3}) - w(T_1) - w_u.\] 
Thus
$\Delta(S_{-u},S'_u) + \Delta(S_{-u_4},S'_{u_4}) =  -w_{u} - w_{u_4} < 0$ and 
one of the two players has an incentive to deviate.

When $d=5$,   we have $\Delta(S_{-u},S'_u)=  w(T_1) + w_{u_2} - w_{u_3} - w(T_{4}) - w_{u_5}$ 
and $\Delta(S_{-u_5},S'_{u_5})= w(T_{4}) +  w_{u_3} - w_u - w(T_1) - w_{u_2}$. 
Therefore we have that 
$\Delta(S_{-u},S'_u) + \Delta(S_{-u_5},S'_{u_5}) =  -w_{u} - w_{u_5} < 0$ and 
one of the two players has an incentive to deviate.
\end{proof}

We need to prove   first an auxiliary result. 

\begin{lemma}
\label{lem:Wtrees}
Let $\Gamma= \langle  V, (w_u)_{u\in V}, \alpha,\beta \rangle$ be a \ourgamen
and let $G$ be a \NE graph  of $\Gamma$. 
If there is  $v\in V$ with  $diam_G(v)\leq \beta-1$, 
then $W(G,\beta) \leq \alpha(n-1)$.
\end{lemma}

\begin{proof}
Let $S\in\NE(\Gamma)$ and let  $G=G[S]$.
Let $u\in V$, $u\neq v$. 
If $v\notin S_u$,  
$\Delta(S_{-u}, S_u\cup \{v\})\geq \alpha -\sum_{\{x | d_G(u,x)> \beta\}}w_x \geq 0$. 
But, if $v\in S_u$,  $diam(u)\leq \beta$.

Hence, $\alpha \geq \sum_{\{x | d_G(u,x)> \beta\}}w_x$ and 
summing over all $u\not = v$ we have that $\alpha(n-1) \geq W(G,\beta)$.
\end{proof}

The proof of the upper bound for the \PoA on \NE trees 
uses the previous statements and examines the particular cases $\beta=2,\, 3$. 

\begin{theorem}\label{theo:poatrees}
The \PoA on  \NE trees of a \ourgamen  is at most $2$. 
\end{theorem}

\begin{proof}
Let $T$ be a \NE tree of $\Gamma$. 
From Proposition~\ref{pro:diamtrees}  we have a bound on the diameter, 
so we know that  $diam(T)\leq \beta +1$. 
Since $T$ is a tree, we have that there exists $u \in V$ such that  
$diam(u)\leq (diam(T)+1)/2\leq \beta/2 + 1$.
If $\beta \geq 4$, then $diam(u)\leq \beta -1$. 
By Lemma~\ref{lem:Wtrees}, $C(T) \leq 2\alpha(n-1)$. 
Hence, the \PoA of \NE trees of $\Gamma$ is at most 2 for $\beta \geq 4$.

In the case of $\beta = 3$, either $diam(T)\leq 3$ or $diam(T)= 4$. 
In the first case $C(T)=\alpha(n-1)$ and in the second there is $u$ 
with $diam_T(u)=2=\beta-1$ and we can use Lemma~\ref{lem:Wtrees}.

Finally, we consider the case $\beta=2$. 
Notice that the unique tree $T$ with diameter $3$ is a double star, 
a graph that is formed by connecting the centers of two stars. 
Assume that a \NE tree $T$ is formed by $T_u$, a star with center $u$, 
and $T_v$, a star graph with center $v$, joined by the edge $(u,v)$.
Let $L_u$ ($L_v$) be the set of  leaves in $T_u$ ($T_v$). As $T$ is a \NE 
graph we have that  $w(L_u),w(L_v)\leq\alpha$. 
Furthermore 
\begin{align*}
C(T) = & \alpha (n-1) + \sum_{w\in L_u} w(L_v) + \sum_{w\in L_v} w(L_u) \leq    \alpha (n-1) + \sum_{w\in L_u} \alpha + \sum_{w\in L_v} \alpha\\
& \leq      \alpha (n-1)  + \alpha (n-2)\leq 2\alpha (n-1).
\end{align*}
\end{proof}

Note that in a \NE tree $T$, if $\alpha >\wM$, 
for an edge $\{u,v\}$ connecting a leaf $u$, it must be the case that $v\in S_u$. 
Then, in the proof of Proposition~\ref{pro:diamtrees}, 
we only have the case $u_1\in S_u$. In such case $diam(T)\leq \beta$. Hence, if  $\alpha >\wM$,
the \PoA on  \NE trees is  $1$.

\begin{corollary}
Let $\Gamma= \langle  V, (w_u)_{u\in V}, \alpha,\beta \rangle$ 
be a \ourgamen such that  $\alpha > \wM$. 
For any  \NE tree of $\Gamma$, $diam(T)\leq \beta$ and therefore the \PoA on  \NE trees is  $1$.
\end{corollary}

To tighten the upper bound  let us analyze the properties of the \NE trees with diameter $\beta+1$

\begin{lemma}\label{lem:tight-tree}
Let $T$ be a \NE tree of $\Gamma= \langle  V, (w_u)_{u\in V}, \alpha,\beta \rangle$  
having $diam(T)=\beta +1$, for some $\beta\geq 3$. 
Let $P=u, u_1, \dots, u_\beta, v$ be a diametral path in $T$ and let $S$ 
be  a \NE such that $T=G[S]$. We have that
\begin{enumerate}
\item $u$ and $v$ are leaves of $T$.
\item $S_u=S_v=\emptyset$.
\item $w(u)=w(v)=\alpha$.
\item $P$ is the unique diametral path in $T$. 
\end{enumerate}
\end{lemma}
\begin{proof}
Statement 1 follows from the fact that $T$ is a tree with diameter $\beta+1$. 

To prove the second statement, assume that $S_u\neq \emptyset$.  
As $u$ is a leaf it must be the case that $S_u=\{u_1\}$. 
Consider the strategy $S'_u=\{u_2\}$. Taking into account  
that $d_T(u_2,v)=\beta-1$ and that the tree rooted at $u_2$ 
after deleting $(u,u_1)$ and $(u_2,u_3)$ has depth at most 2, 
we have that $\Delta(S_{-u},S'_u)\leq -w(v) <0$. 
Contradicting the fact that $T$ is a \NE tree.  
A symmetric argument shows that $S_v=\emptyset$.

To prove the third statement we consider two cases.

\smallskip
\noindent
\emph{Case 1:} $w(u)>\alpha$.  Let $S'_v = \{u_1\}$, 
then $\Delta(S_{-v},S'_v) \leq \alpha - w(u) < 0$. Thus $T$ could not be a \NE.

\smallskip
\noindent
\emph{Case 2:} $w(u)<\alpha$.  
By 2 we know that $S_u=S_v=\emptyset$, therefore $u\in S_{u_1}$. 
Taking $S'_{u_1}=S_{u_1}\setminus \{u_1\}$ we have 
$\Delta(S_{-u_1},S'_{u_1}) \leq w(u)- \alpha<0$. 
Again $S$ could not be a \NE.

We conclude that $w(u)=\alpha$.  A symmetric argument shows that $ w(v) = \alpha$.

To prove the last statement assume that $T$ has two  
diametral paths with length $\beta+1$. Let $u,v,u',v'$ be four vertexes 
such that $d(u,v)=d(u',v')=\beta+1$.
We consider two cases. 

\smallskip
\noindent
\emph{Case 1:} the four vertexes are different.   
Let $P$ be the shortest path from $u$ to $v$ and $P'$ the 
shortest path from $u'$ to $v'$.  Let us first show that $P$ and $P'$ 
must share at least one point.  
Otherwise let $y$ be the vertex  in $P$ 
that is closest to $P'$ and  let $x$ be the vertex in $P'$ that is closest to $y$. 
By construction $P'$ lies in the subtree rooted at $y$ after removing the edges 
in $P$, thus $d(y,x)>0$. 
Therefore, $\max\{d(u,y),d(y,v)\}+d(x,y)+\max\{d(u',x),d(x,v')\} > \beta+1$. 
Contradicting the fact that $T$ has diameter $\beta+1$. 

Thus $P$ and $P'$ share at least one point. 
Let $x(y)$ be the vertex common to $P$ and $P'$ that is closer to $u(v)$.  
If there is only one common point $x=y$. 
Observe that when $x=y$ it must happen that $x$ is the central point of both paths, 
that is $\beta+1$ must be even and $d(u,x)=d(v,x)=d(u'x)=d(v'x)= (\beta+1)/2$. 
When $x\neq y$ assume without loss of generality that $u'$ is the vertex 
in the subtree rooted at $x$ after removing $P$. 
In such a case, $d(u',x)=d(u,x)\leq  (\beta+1)/2$ 
and  $d(v,y)=d(v',y)\leq (\beta+1)/2$ as otherwise the tree 
will not have diameter $\beta+1$.   
Thus $d(u,v')=\beta+1$.  
By 2 we know that $S_u=\emptyset$ and by 3 that $w(v)=w(v')=\alpha$. 
Consider the strategy profile, $S'_u=\{y\}$. 
We have that   $\delta(S_{-u},S'_u)\leq \alpha - w(v)-w(v')<0$. Therefore $T$ cannot be a \NE.

\smallskip
\noindent
\emph{Case 2:} two vertexes are the same. Without loss of generality assume that $u'=u$. 
Let $y$ be the branching point of the paths from $u$ to $v$ and $u$ to $v'$. 
As in the previous case, we have that $d(y,v)=d(y,v')\leq (\beta+1)/2$. 
Considering $S'_u=\{y\}$ we have again that   $\delta(S_{-u},S'_u)\leq \alpha - w(v)-w(v')<0$. 
Therefore $T$ cannot be a \NE. 

\smallskip
We conclude that there are only two vertexes at distance $\beta+1$ in $T$.      
\end{proof}

Putting all together we get an upper bound on the \PoA on \NE trees when $\beta\neq 2$.
\begin{theorem}\label{theo:11}
The \PoA on  \NE trees of a \ourgamen  with $\beta\geq 3$ and $n$ players is at most $1+\frac{2}{n-1}$. 
\end{theorem}
\begin{proof}
For \NE trees with diameter $\leq \beta$ the social cost is $\alpha(n-1)$ 
but for \NE with diameter $\beta+1$, by Lemma~\ref{lem:tight-tree}, 
the social cost is $\alpha(n-1) + 2\alpha$. 
As a star is an optimal graph with social cost $\alpha(n-1)$ the claim follows.
\end{proof}

\begin{figure}
\begin{center}
\begin{tikzpicture}[every node/.style={circle,scale=0.8}, node distance=7mm, >=latex]
\node[draw](a) at (0,4) {$\alpha$};
\node[draw](b) at (2,4) {$w_1$};
\node[draw](c) at (4,4) {$w_2$};
\node[draw](d) at (6,4) {$\alpha$};

\node[draw](a1) at (0,2) {$\alpha$};
\node[draw](b1) at (2,2) {$w_1$};
\node[draw](c1) at (4,2) {$w_2$};
\node[draw](d1) at (6,2) {$\alpha$};

\draw[->] (b) to node {}(a);
\draw[->] (c) to node {}(b);
\draw[->] (c) to node {}(d);
\draw[->] (b1) to node {}(a1);
\draw[->] (b1) to node {}(c1);
\draw[->] (c1) to node {}(d1);

\end{tikzpicture}
\end{center}
\caption{The \NE trees with diameter 3\label{fig:P3}}
\end{figure}
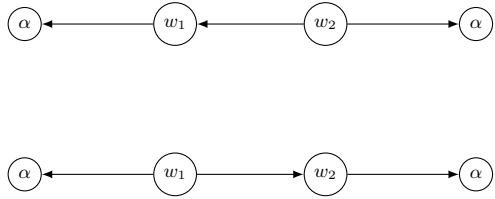

For the case $\beta=2$ it remains to analyze whether a double star can be a \NE for a \ourgamen.   

\begin{lemma}\label{lem:extra2}
Let $T$ be a \NE tree of  a \ourgamen 
$\Gamma= \langle  V, (w_u)_{u\in V}, \alpha,\beta \rangle$  let $\beta=2$. 
There is no \NE tree for $\Gamma$ with diameter 3 except when $|V|=4$ 
and at  least two players have weight $\alpha$.

\end{lemma}
\begin{proof}
Assume that a double star $T$  is formed by 
two starts $T_u$ and $T_v$ with centers $u$ and $v$ 
respectively and  the edge $\{u,v\}$.   
Let $L_u$ ($L_v$)  be the set of leaves in $T_u$ ($T_v$).  
Let $S$ be a \NE so that $T=G[S]$.   
As $T$ is a \NE we know that $w (L_u), w(L_v)\leq \alpha$, 
otherwise by  connecting a leaf to the other center their cost will decrease.   

Assume that $|L_u|,|L_v| \geq 2$. 
We have that,  for a leaf  $x$, $w(x)<\alpha$.  
So, $u\in S_x$, for $x\in L_u$, and $v\in S_y$, for $y\in L_v$. 
Otherwise, $u$ (or $v$) would benefit by disconnecting to their leaves.  
For any leaf $x\in L_u$ ($y\in L_v$), consider the strategy $S'_x=\{v\} $ ($S'_y=\{u\} $).  
For  $x\in L_u$, we have $\Delta(S_{-x}, S'_x)= w(L_u) -w(x) - w(L_v)\geq 0$, 
that is $w(x) \leq w(L_u)-w(L_v)$. 
For  $y\in L_v$, we have $\Delta(S_{-y}, S'_y)= w(L_v) -w(y) - w(L_u)\geq 0$, 
thus $w(y) \leq w(L_v)-w(L_u)$.  Which is impossible as the node weights are positive.  
Therefore $|L_u|= 1$ or $|L_v|=1$.  

Let us assume w.l.o.g that $L_u=\{x\}$.  
If  $w(x)<\alpha$ and  $x\in S_u$,  $\Delta(S_u,\emptyset) = w(x) - \alpha<0$, 
which is not possible. 
Therefore,  $u\in S_x$. But in such a case  $\Delta(S_{-x}, \{v\}) = - w(L_v) < 0$.    
So, $w(x) = \alpha$. 

If $|L_v|>1$,  let $y\in L_v$. As $w(L_v) \leq \alpha$ and $w(y)>0$, we have $w(y)<\alpha$. 
Therefore, $v\in S_y$, but then 
$$\Delta(S_{-y}, \{u\}) = -w(x)  +  w(L_v)-w(y) = -\alpha +  w(L_v)-w(y) < 0.$$ 
Contradicting that $S$ is a \NE. Thus, $L_v=\{y\}$ and, as for the case $L_u=\{x\}$, 
we can conclude that $w(y)=\alpha$. 

The unique  graph satisfying all conditions is a path on 4 vertexes. 
Furthermore the leaf nodes must have weight $\alpha$ and there 
are no restrictions for the  weights of the internal vertexes. 
It is easy to see that the unique orientations producing a 
\NE in this particular case are  the ones depicted in Figure~\ref{fig:P3}. 
\end{proof}

\begin{theorem}\label{theo:poatrees}
The \PoA  on \NE trees of \ourgamens   is $\leq 5/3$ and 
there are games for which a \NE  tree has cost $5 \, \optv/3$.   
\end{theorem}
\begin{proof}
For $\beta\geq 3$ and $n\geq 4$, the \PoA on \NE trees is at most $1+ \frac{2}{n-1}\leq 5/3$, 
by Theorem~\ref{theo:11}.
For  $\beta\geq 3$ and $n< 4$, all trees have diameter at most $\beta$, so  the \PoA on \NE trees is 1. 
For $\beta=2$ according to Lemma~\ref{lem:extra2} 
all \NE trees have diameter at most $\beta$ except for $P_3$ in some cases. 
When $P_3$ is a \NE we have that $C(P_3)=3\alpha + 2\alpha = 5\alpha$, 
giving the upper bound. As there are games for which $P_3$ 
is a \NE (see   Figure~\ref{fig:P3}), the claim follows.
\end{proof}

\section{Celebrity games for  $\beta = 1$}
\label{sec:beta1}
Let us now analyze the case  $\beta=1$. 
Observe that every player $u$ for each non-adjacent node  
$v$  pays  $w_v$, and  for each adjacent node   pays either 
$\alpha$  if he has bought the link,  or  $0$, otherwise.   
Notice that if $u$  establishes the link $(u,v)$, only the node $v$ 
will take profit of this decision.  Contrasting with this, when $\beta>1$, 
if player $u$  pays  a new link,  then all the nodes that get closer to $u$ 
but  not farther than $\beta$,  will take advantage of this new link.

This particular behavior allows us to  show that computing a best response 
becomes a tractable problem. Furthermore, the structure of \NE and \OPT  
graphs is quite different from the case of $\beta>1$ and we can obtain a tight bound for the \PoA.  

\begin{proposition}
\label{prop:Best-Response-2}   
The  problem of computing a best response of a player to a strategy profile in  \ourgames 
is  polynomial time solvable when $\beta = 1$.
\end{proposition}

\begin{proof}
Let $S$ be a strategy profile of $\Gamma=\langle  V, (w_u)_{u\in V}, \alpha,1 \rangle$ and let $u\in V$.  
Consider another strategy profile $S'=(S_{-u}, S'_u)$, for some $S'_u\subseteq V\setminus\{u\}$. 
As $\beta=1$  we have
\[c_u(S')= \alpha |S'_u| + \sum_{v\notin S'_u} w_v.\]
Note that, when $|S'_u|=k$, the first component of the cost is the 
same and thus a best response on strategies with $k$ players can 
be obtained by taking from $S'_u$ the players with the $k$-th highest weights. 
Let $S'_u(k)$ be the set of those players and let $W_k= W-w(S'_u(k))$. 
Thus $c_u(S_{-u},S'_u(k))= \alpha k + W_k$.
To obtain a best response it is enough to compute the  value $k$ 
for which $c_u((S_{-u},S'_u(k)))$ is minimum and output $S'_u(k)$. 
Observe that the overall computation can be performed in polynomial time.  
\end{proof}

In order to show  a bound for the \PoA  we prove first some auxiliary results.
When  $\beta=1$  pairs of vertexes at distance bigger than one correspond 
to pairs of vertexes that are not connected by an edge and such a property 
does not hold for higher values of $\beta$.

\begin{proposition}\label{b1:ne}
Let $\Gamma=\langle  V, (w_u)_{u\in V}, \alpha,1 \rangle$ be a \ourgame. 
If  $G=(V,E)$  is a \NE graph of $\Gamma$,  for each $u,v\in V$, 
\begin{itemize} 
\item if either $w_u > \alpha$ or $w_v>  \alpha$  then $\{u,v\} \in E$,
 \item  if both $w_u < \alpha$ and $w_v <  \alpha$  then $\{u,v\} \notin E$,
\item otherwise  the edge  $\{u,v\}$ might or might not belong to $E$.
\end{itemize}
 \end{proposition}
\begin{proof}
Let $S$ be a \NE and  let $G=G[S]=(V,E)$. 
Observe that due to the fact that $\beta=1$, for any player $u$,
\[c_u(S)= \alpha |S_u| + \sum_{\{v \mid v\neq u, \{u,v\}\not\in E\}} w_v.\]
The cost is thus expressed in terms of the existence or non existence 
of a connection between pairs of nodes and thus the strategy can be 
analyzed considering  only deviations in which a single edge is added or removed.
We analyze the different cases for players $u$ and $v$.

\smallskip
\noindent
\emph{Case 1:  $w_u>\alpha$}. For any player $v\neq u$,   
if the edge $\{u,v\}$ is not present in $G$ the graph cannot be a 
\NE graph as $v$  improves its cost by connecting to $u$.  
For the same reason,  if the edge is present either  $u\in S_v$ or $v\in S_v$.  
The latter case$v\in S_v$, can happen only when $w_v>\alpha$.  
Therefore,  the player that is  paying for the connection will 
not obtain any benefit by deviating. 

\smallskip
\noindent
\emph{Case 2:  $w_u,w_v < \alpha$}. If the edge $\{u,v\}$ is present in $G$ 
the graph cannot be a \NE  graph as  the player establishing the connection 
improves its cost by removing the connection to the other player. 
For the same reason, if the edge is not present none of the  players will  
obtain any benefit by deviating and paying for the connection.

\smallskip
\noindent
\emph{Case 3: $w_u,w_v = \alpha$}. The cost, for any of the players, 
of establishing the connection or not is the same.  
In consequence the edge can or cannot be in a \NE graph.

\smallskip
\noindent
\emph{Case 4:  $w_u=\alpha$ and $w_v< \alpha$}. 
Player $v$ is indifferent to be or not to be connected to $u$, 
but player $u$ in a \NE  will never include  $v$ in its strategy. 
Observe that again the edge can or cannot exists in a \NE graph but, 
if it exists,  it can only be the case that $u\in S_v$.
\end{proof}


Let us analyze now the structure of the \OPT graphs.

\smallskip
\noindent
\begin{proposition} \label{b1:opt}
 Let  $G=(V,E)$  be a \OPT graph of a 
 \ourgame $\Gamma=\langle  V, (w_u)_{u\in V}, \alpha,1 \rangle$.
For any $u,v\in V$, we have 
\begin{itemize} 
\item if  $w_u + w_v  <\alpha$ then    $\{u,v\} \notin E$,
 \item  if $w_u + w_v  > \alpha$ then    $\{u,v\} \in E$,
\item if $w_u+w_v=\alpha$ then    $\{u,v\}$ might or not be an edge in $G$.
\end{itemize}
\end{proposition}
\begin{proof}
Let $S$ be a strategy profile  and  let $G=G[S]=(V,E)$ be an \OPT graph. 
As we have seen before as $\beta=1$, for any player $u$,
\[c_u(S)= \alpha |S_u| + \sum_{\{v\mid v\neq u, \{u,v\}\not\in E\}} w_v,\]
and we get et the following expression for the social cost
\[C(G)=  \alpha |E|+ \sum_{\{u,v \mid u<v, \{u,v\} \not\in E\}}(w_u + w_v).\]
The above expression shows that to minimize the contribution to the cost,  
an edge $ \{u,v\}$  can be present in the graph only if  $w_u+w_v\geq \alpha$ 
and will appear for sure only when $w_u+w_v> \alpha$. Thus the claim follows. 
\end{proof}

From the previous characterizations we can derive a constant upper bound 
for the price of anarchy when $\beta=1$. 

\begin{theorem}\label{teo:beta1}
 Let $\Gamma=\langle  V, (w_u)_{u\in V}, \alpha,1 \rangle$ be a \ourgame.
$\PoA(\Gamma) \leq 2$.  
Furthermore the, ratio among the social cost of the best 
and the worst \NE graphs of $\Gamma$ is bounded by 2. 
\end{theorem}

\begin{proof}
$\Gamma=\langle  V, (w_u)_{u\in V}, \alpha,1 \rangle$.  
Observe that due to the conditions given in 
Propositions~\ref{b1:ne} and~\ref{b1:opt} 
the social cost of an \OPT graph is   
\[\sum_{\{\{u,v\}\mid w_u+w_v\geq \alpha\}} \hspace{-24pt} \alpha \ \  \ \  \ +  
\sum_{\{\{u,v\}\mid w_u+w_v < \alpha\}}  \hspace{-24pt} (w_u+w_v) ,\]
and the social cost of a \NE graph with minimum  number of edges, i.e., 
one in which all the  optional are not present,  is at most
\begin{align*}
 \sum_{\{\{u,v\}\mid w_u> \alpha \text{ or }  w_v> \alpha\}} \alpha  &+  
 \sum_{\{\{u,v\}\mid w_u,  w_v \leq \alpha\}}  (w_u+w_v) =\\
  =  & \sum_{\{\{u,v\}\mid w_u> \alpha \text{ or }  w_v> \alpha\}} \alpha\\
      &  + \sum_{\{\{u,v\}\mid w_u,  w_v \leq \alpha \text{ and } w_u+w_v = \alpha \}} \alpha \\
      &  + \sum_{\{\{u,v\}\mid w_u,  w_v \leq \alpha \text{ and } w_u+w_v < \alpha \}}  (w_u+w_v)\\
      & + \sum_{\{\{u,v\}\mid w_u, w_v  \leq  \alpha  \text{ and }  w_u+w_v > \alpha \}}  (w_u+w_v)
\end{align*}
Observe that  the difference with the cost of an \OPT graph is in the last term 
 \[D= \{\{u,v\}\mid  w_u, w_v \leq \alpha \text{ and } w_u + w_v >  \alpha\}.\] 
Notice that $\{u,v\}\in D$ contributes to the cost of an \OPT 
graph with $\alpha$ and to the cost of a \NE graph with $w_u + w_v$.
By taking $\Gamma$ with $w_u=\alpha$, for any $u\in V$, 
we can maximize the size of $D$ and this leads to the worst possible \NE graph. 
For such a $\Gamma$, $I_n$ is a \NE graph 
and  we have that $C(I_n)=\sum_{u,v\in V, u<v} (w_u +w_v) =\alpha n(n-1)$.  
Furthermore,  in any \OPT graph of $\Gamma$, all the edges will be present, 
thus we have $\OPT = \alpha n(n-1)/2$. Thus 
\[\PoA (\Gamma) \leq \frac{ n(n-1)\alpha}{\alpha n (n-1)/2} = 2 \]  
Observe that when  $w_u=\alpha$, for any $u$,  
the complete graph is also a \NE graph and thus we have that 
the ratio between the social cost of the worst  and the best \NE graph is bounded by 2.   
\end{proof}

Observe that when $\alpha < w_{min}$  the unique \NE is a 
complete graph which is also an \OPT graph. 
Taking into account that the relationship among \ourgames and \maxbds 
provided in Proposition~\ref{CelebritiesVSMaxBD} 
also holds for $\beta=1$ we can conclude. 

\begin{corollary} For $\beta=1$, the \PoA and the \PoS of 
\maxbds and  \ourgames with $\alpha < w_{min}$  is 1. 
\end{corollary}

\section{Conclusions and Open Problems}
\label{sec:conclu}

We have introduced  the  \ourgames model aiming to address 
the creation of networks in a scenario where the nodes or 
players may have different weights and where the requirement 
of being close to  a global critical distance has to be balanced against  the node weights.
Our results provide further  understanding of the structural properties 
of stable networks. 
We have  shown that the critical distance affects directly  the diameter of the stable networks. 
For \stargames the diameter is $\leq 2\beta + 1$  and, 
in the case that the \NE graph is not a tree, the diameter is  $> \beta/2$. 
Furthermore, this critical distance, jointly with  player weights and link establishment cost, 
have implications on the quality of the \NE.  
We have shown that the  \PoA of \ourgamens\ is $O(\min\{n/\beta,W/\alpha\})$ and, 
for $\beta=2$, we have found games whose \PoA is $\Omega(n)$.  
In contra-position restricting the \NE to be trees the \PoA is constant.

We can observe that,  as one can expect,  
enlarging the value of the critical distance improves the quality 
of equilibria. 
Furthermore, if the total  game weight  $W=O(\alpha)$, 
the \PoA is $O(1)$. Corresponding to the intuition that when  
player's weights are negligible players prefer to be isolated.  
In contrast,  when  all the players are celebrities, even though 
their weights could be very different, players prefer to be closer, 
and the \NE graphs have diameter $\leq \beta$. 
In this latter case, the upper bound on the \PoA  
obtained in~\cite{BiloGP:15} for \maxbds ameliorates 
the upper bound of \ourgames. 

It still remains open to  shorten the gap between
the lower and upper bounds on the \PoA. 
Our results are only tight for $\beta=1$ and $\beta=2$. 
The cases where $\beta$ is constant are of particular interest.  
In the family of graphs providing the lower bound on the \PoA 
not all the nodes are celebrities, so our result has no implication for \maxbds. 

Further questions of interest are to study natural variations of our framework. 
Among the many possibilities, we propose
 to analyze \ourgames under
(i)  the Max cost model (work in progress),
(ii) other definitions of the social cost.

Finally, we have not considered the non uniform version 
where each player $u$ can have its own critical distance $\beta_u$. 
\cite{BiloGP:15} showed that the \PoA of \maxbd is $\Omega(n)$ even 
for the non uniform model with only two distance-bound values. 
As we have mentioned before such a negative result for \maxbds  
translates to the \ourgames when all the players are celebrities. 

\section*{Acknowledgements}
We thank anonymous reviewers for all their useful comments and suggestions, 
they helped, undoubtedly, to improve the quality of this paper.

This work was partially
supported by funds from the AGAUR of the 
Government of Catalonia under project ref. SGR 2014:1034 (ALBCOM). 
C. \`{A}lvarez, A. Duch  and M. Serna were partially supported by the 
Spanish Ministry for Economy and Competitiveness (MINECO) and 
the European Union (FEDER funds), under grants ref. TIN2013-46181-C2-1-R (COMMAS). 
M. Blesa was partially supported by the Spanish Ministry for 
Economy and Competitiveness (MINECO) and the European Union (FEDER funds), 
under grant TIN2012-37930-C02-02.


\end{document}